\documentclass[conference]{IEEEtran}
\IEEEoverridecommandlockouts

\usepackage{algorithm,algpseudocode,tabularx,diagbox}
\usepackage{paralist}
\usepackage{placeins}

\usepackage[utf8]{inputenc}
\usepackage{acronym,amsmath,amssymb,amsthm,dsfont,bbding,fancyhdr,lastpage} 
\usepackage[only,llbracket,rrbracket]{stmaryrd} 
\usepackage{paralist}

\RequirePackage{xcolor}
\definecolor{darkgray}{rgb}{0.66, 0.66, 0.66}
\definecolor{darkblue}{HTML}{4169E1}

\RequirePackage{graphicx}

\newtheorem{theorem}{Theorem}%
\newtheorem{corollary}[theorem]{Corollary}%
\newtheorem{lemma}[theorem]{Lemma}%
\newtheorem{definition}[theorem]{Definition}%

\acrodef{ACDIS}[ACDIS]{Adaptive Communication Decision and Information Systems}
\acrodef{AEP}{asymptotic equipartition property}
\acrodef{AoA}{Angle of Arrival}
\acrodef{AWGN}{Additive White Gaussian Noise}
\acrodef{AVC}[AVC]{Arbitrarily Varying Channel}
\acrodef{BER}{Bit-Error-Rate}
\acrodef{BEC}{Binary Erasure Channel}
\acrodef{BPSK}{Binary Phase-Shift Keying}
\acrodef{BSC}{Binary Symmetric Channel}
\acrodef{BICM}[BICM]{Bit-Interleaved Coded-Modulation}
\acrodef{BP}[BP]{Basis Pursuit}
\acrodef{CDF}[CDF]{Cumulative Distribution Function}
\acrodef{CGF}[CGF]{Cumulant Generating Function}
\acrodef{CLT}[CLT]{Central Limit Theorem}
\acrodef{CSI}[CSI]{Channel State Information}
\acrodef{CQ}[CQ]{Classical-Quantum}
\acrodef{DMC}[DMC]{Discrete Memoryless Channel}
\acrodef{DMS}[DMS]{Discrete Memoryless Source}
\acrodef{ERM}[ERM]{Empirical Risk Minimization}
\acrodef{EPR}[EPR]{Einstein-Podolsky-Rosen}
\acrodef{FER}[FER]{Frame Error Rate}
\acrodef{ICA}[ICA]{Independent Component Analysis}
\acrodef{iid}[i.i.d.]{independent and identically distributed}
\acrodef{IoT}[IoT]{Internet of Things}
\acrodef{JCS}[JCS]{Joint Communication and Sensing}
\acrodef{KKT}[KKT]{Karush-Kuhn Tucker}
\acrodef{LASSO}[LASSO]{Least Absolute Shrinkage and Selection Operator}
\acrodef{LP}[LP]{Linear Program}
\acrodef{LPD}[LPD]{Low Probability of Detection}
\acrodef{LDPC}[LDPC]{Low-Density Parity-Check}
\acrodef{LLMS}[LLMS]{Linear Least Mean Square}
\acrodef{LMS}[LMS]{Least Mean Square}
\acrodef{MAC}[MAC]{multiple-access channel}
\acrodef{MGF}[MGF]{Moment Generating Function}
\acrodef{MLC}[MLC]{Multi-Level Coding}
\acrodef{MLE}[MLE]{Maximum Likelihood Estimate}
\acrodef{MIMO}[MIMO]{Multiple-Input Multiple-Output}
\acrodef{MISO}{Multiple-Input Single-Output}
\acrodef{MSD}[MSD]{Multi-Stage Decoding}
\acrodef{MMSE}[MMSE]{Minimum Mean-Square Error}
\acrodef{OTP}[OTP]{One-time Pad}
\acrodef{PAC}[PAC]{Probably Approximately Correct}
\acrodef{PCA}[PCA]{Principal Component Analysis}
\acrodef{PDF}[PDF]{Probability Density Function}
\acrodef{PMF}[PMF]{Probability Mass Function}
\acrodef{POVM}[POVM]{Positive Operator-Valued Measure}
\acrodef{PPM}[PPM]{Pulse Position Modulation}
\acrodef{PSD}{Power Spectral Density}
\acrodef{PSK}{Phase Shift Keying}
\acrodef{QKD}{Quantum Key Distribution}
\acrodef{QMF}[QMF]{Quantum-Memory-Free}
\acrodef{QSDC}{Quantum Secure Direct Communication}
\acrodef{ROC}{Receiver Operating Characteristic}
\acrodef{CVQKD}{Continuous-Variable \ac{QKD}}
\acrodef{QPSK}{Quadrature Phase-Shift Keying}
\acrodef{RV}{random variable}
\acrodef{SIMO}{Single-Input Multiple-Output}
\acrodef{SNR}{Signal-to-Noise Ratio}
\acrodef{SVM}[SVM]{Support Vector Machine}
\acrodef{TPCP}{Trace-Preserving Completely-Positive}
\acrodef{wrt}[w.r.t.]{with respect to}
\acrodef{WSS}{Wide Sense Stationary}
\acrodef{i.i.d.}{independent and identically distributed}
\acrodef{FDG}{functional dependence graph}
\acrodef{PBS}{polarization beamsplitter}
\acrodef{PD}{photodetector}
\acrodef{CPTP}{completely positive trace-preserving}

\newcommand{\calB}{\mathcal{B}}

\newcommand{\calD}{\mathcal{D}}

\newcommand{\calE}{\mathcal{E}}

\newcommand{\calH}{\mathcal{H}}

\newcommand{\calL}{\mathcal{L}}

\newcommand{\calN}{\mathcal{N}}

\newcommand{\calW}{\mathcal{W}}

\newcommand{\norm}[2][]{{\left\Vert{#2}\right\Vert}_{#1}}

\newcommand{\argmax}{\mathop{\text{argmax}}}

\renewcommand{\P}[2][]{{\mathbb{P}_{#1}\left(#2\right)}}

\newcommand{\bra}[1]{{\left\langle{#1}\right\vert}}
\newcommand{\ket}[1]{{\left\vert{#1}\right\rangle}}

\newcommand{\D}[2]{{\mathbb{D}}\!\left(#1\,\middle\Vert\,#2\right)} 
\newcommand{\trD}[3][\alpha]{{\widetilde{\mathbb{D}}_{#1}}\!\left(#2\,\middle\Vert\,#3\right)} 

\newcommand{\V}[2]{{\mathbb{V}}\!\left(#1\,\middle\Vert\,#2\right)} 

\newcommand{\I}[2]{{{\mathbb{I}}\!\left(#1;#2\right)}} 
\newcommand{\vonH}[1]{{\mathbb{H}}\!\left(#1\right)}
\newcommand{\rH}[2][\alpha]{\mathbb{H}_{#1}\left(#2\right)} 
\newcommand{\smH}[2][\epsilon]{\mathbb{H}_{\text{\textnormal{min}}}^{#1}\left(#2\right)} 
\newcommand{\sMH}[2][\epsilon]{\mathbb{H}_{\text{\textnormal{max}}}^{#1}\left(#2\right)} 

\newcommand{\avgI}[1]{{{\mathbb{I}}\!\left(\smash{#1}\right)}} 

\renewcommand{\leq}{\leqslant} 
\renewcommand{\geq}{\geqslant} 

\newcommand{\tr}[2][]{\ensuremath{\text{\textnormal{tr}}_{#1}\left(#2\right)}}  

\newcommand{\Pdist}[1]{\ensuremath{P}\left(#1\right)} 
\newcommand{\ball}[2][]{\calB^{#1}\left(#2\right)} 

\allowdisplaybreaks

\newcommand{\cmnt}[1]{\textcolor{darkblue}{#1}}
\renewcommand{\cmnt}[1]{#1}

\begin{document}
\title{A \acl{QMF} \acl{QSDC} Protocol Based on Privacy Amplification of Coded Sequences}

\author{
    \IEEEauthorblockN{
        Shang-Jen Su        \IEEEauthorrefmark{1},
        Shi-Yuan Wang\thanks{\cmnt{S.-Y. Wang contributed to this work while pursuing the Ph.D. degree at Georgia Institute of Technology.}}       \IEEEauthorrefmark{2}
        Matthieu R. Bloch   \IEEEauthorrefmark{1}}
    \IEEEauthorblockA{
        \IEEEauthorrefmark{1}
        	School of Electrical and Computer Engineering, Georgia Institute of Technology, 
        	Atlanta, GA 30332, USA}
        \IEEEauthorblockA{\IEEEauthorrefmark{2}
      		Qualcomm Technologies Inc.,
      		San Diego, CA 92121, USA}
    \IEEEauthorblockA{
    Email: ssu49@gatech.edu,
    shiyuanw@qti.qualcomm.com,
    matthieu.bloch@ece.gatech.edu}
}

\addtolength{\topmargin}{+0.07in}

\maketitle

\begin{abstract}
We develop an information-theoretic analysis of \ac{QMF} \ac{QSDC} under collective attacks as an alternative to the use of a conventional \ac{QKD} protocol in conjunction with one-time pads. Our main contributions are:
  \begin{inparaenum}[1)]
  \item a \ac{QMF}-\ac{QSDC} protocol that only relies on universal hashing of coded sequences without wiretap coding;
  \item a set of privacy amplification theorems for extracting secrecy from \emph{coded} classical sequences against quantum side-information.
  \end{inparaenum}
These tools open the way to the design of effective \ac{QMF}-\ac{QSDC} protocols.
\end{abstract}

\section{Introduction}
\label{sec:introduction}

The theory and practice of \acf{QKD}~\cite{Bennett1984QuantumCryptography,Bennett1992ExperimentalQuantum} have made tremendous progress. Modern proofs now account for finite-length effects in the operation of the protocols and offer measurement-device-independent security guarantees~\cite{Scarani2009SecurityPractical, Pirandola2020AdvancesQuantum, Xu2020SecureQuantum}. Record distances and rates have also been achieved, paving the way for quantum-secure communications at scale~\cite{Liao2018SatelliterelayedIntercontinental,Chen2021IntegratedSpacetoground}. Unfortunately, standard \ac{QKD} protocols often remain inefficient:
\begin{inparaenum}[1)]
\item protocols inefficiently use quantum resources, especially when requiring a sifting phase to reconcile preparation and measurement bases;
\item protocols focus on the key generation and defer the task of securing a transmission to subsequent protocols, typically by encrypting data using one-time pads.
\end{inparaenum}

This state of affairs has prompted two major developments. First is the investigation of more efficient \ac{QKD} protocols, exemplified by~\cite{Long2002TheoreticallyEfficient}, which rely on entanglement and quantum memories to more efficiently consume \ac{EPR} pairs than other \ac{QKD} protocols such as~\cite{Ekert1991QuantumCryptography}.
Second is the study of \acf{QSDC} protocols, exemplified \cmnt{by~\cite{Qi2019ImplementationSecurity,Wu2019SecurityQuantum,Deng2004SecureDirect,Ying2024PassiveDecoystate,Paparelle2025ExperimentalDirect}}, which attempt to secure the reliable communication of messages without explicitly relying on a secret-key generation phase. \ac{QSDC} is motivated by the vast literature on wiretap coding in classical and quantum information-theoretic secrecy~\cite{Wyner1975WiretapChannel,Csiszar1978BroadcastChannels,Devetak2005PrivateClassical,Hayashi2015QuantumWiretap}, which shows that channels can be effectively used to \emph{jointly} ensure reliability and secrecy, instead of separately~\cite{Ahlswede1993CommonRandomness,Maurer1993SecretKey}. Recent advances in~\ac{QSDC} include proposals for~\ac{QMF}-\ac{QSDC} protocols that circumvent the current technological bottlenecks of quantum memories~\cite{Sun2018DesignImplementation,Sun2020PracticalQuantum}.

Assessing the full benefits of \ac{QSDC} over \ac{QKD} requires careful inspection. Specifically, the secure rate of~\ac{QKD} and~\ac{QSDC} can only be assessed \emph{after} a channel estimation phase. Since~\ac{QKD} solely focuses on generating random secret keys that need not be known beforehand, post-processing in the form of reconciliation~\cite{Brassard1993SecretkeyReconciliation,Martinez-Mateo2013KeyReconciliation} and privacy amplification~\cite{Bennett1995GeneralizedPrivacy} is performed \emph{a posteriori} once the channel quality is known. In particular, reconciliation and privacy amplification can be made universal.\footnote{Strictly speaking, only privacy amplification is universal; however, in practice, reconciliation algorithms are flexible enough to adapt to varying channel conditions~\cite{Elkouss2009EfficientReconciliation}.} For~\ac{QSDC} with quantum memories, the joint encoding for reliability and secrecy with a wiretap code can, in principle, be deferred to after the estimation phase. However, with the exception of~\cite{Bellare2012SemanticSecurity,Chou2022ExplicitWiretap,Wu2022QuantumSecure}, many wiretap code constructions are not universal~\cite{Thangaraj2007ApplicationsLDPC,Mahdavifar2011AchievingSecrecy,Rathi2013PerformanceAnalysis,Chou2016PolarCoding}. The situation is even more challenging for~\ac{QMF}-\ac{QSDC} since the joint encoding for reliability and secrecy should not only be universal but also happen \emph{a priori} before channel estimation. This situation is similar to the one encountered in wiretap coding over channels with uncertainty, such as wireless channels, for which solutions have been proposed~\cite{Bloch2008WirelessInformationtheoretic,Gungor2013SecrecyOutage,Tahmasbi2020LearningAdversarys}. In brief, the main idea is to avoid leakage by provisioning secret keys ahead of time and systematically one-time padding messages. Secret keys are re-generated \emph{a posteriori}, once channel conditions allow, and added to a pool for later consumption; this operation requires key generation from coded sequences, which is more challenging than from uncoded (\ac{i.i.d.}) ones.

The main contributions of the present work are:
\begin{inparaenum}[1)]
	\item	the proposal of a \ac{QMF}-\ac{QSDC} protocol that does not rely on wiretap coding, and
	\item	an approach that employs a code for reliability and universally extracts secret keys from codewords, which are coded sequences.
\end{inparaenum}

The remainder of the paper is organized as follows. Section~\ref{sec:notation} introduces the notation. We present our system model and protocol in Section~\ref{sec:system-and-protocol}. And we analyze the extracted key length of the protocol in Section~\ref{sec:key-rate-analysis}. Additional details can be found in the arXiv preprint~\cite{Su2026QuantumMemoryFreeQuantum}.

\section{Notation}
\label{sec:notation}
	 Let $\calH_A$ denote the finite-dimensional Hilbert space corresponding to system $A$.
	 Quantum states in $\calH_A$ are represented as normalized density operators in $\calD_{\circ}(\calH_A)\triangleq\{\rho_A : \rho_A \succeq 0 \; \wedge \tr{\rho_A}=1 \}$.
	 Subnormalized quantum states in $\calH_A$ are represented as subnormalized density operators in $\calD_{\bullet}(\calH_A)\triangleq\{\rho_A : \rho_A \succeq 0 \; \wedge \tr{\rho_A} \leq 1 \}$.
	 For $\rho,\sigma \in \calD_{\bullet}(\calH_A)$, the generalized trace distance between $\rho$ and $\sigma$ is denoted by $\frac{1}{2}\norm[1]{\rho-\sigma}+\frac{1}{2}|\tr{\rho-\sigma}|$, where $\norm[1]{\sigma}\triangleq\tr{|\sigma|}$ and $|\sigma|=\sqrt{\sigma^\dagger\sigma}$.
	 The generalized fidelity between $\rho$ and $\sigma$ is defined as $F(\rho, \sigma) \triangleq\left(\|\sqrt{\rho} \sqrt{\sigma}\|_1+\sqrt{(1-\operatorname{tr} \rho)(1-\operatorname{tr} \sigma)}\right)^2$.
	 The purified distance between $\rho$ and $\sigma$ is defined as $\Pdist{\rho,\sigma}\triangleq\sqrt{1-F(\rho, \sigma)}$.
	 Let $\ball[\epsilon]{\rho}\triangleq\{\xi \succeq 0:\Pdist{\rho,\xi}\leq\epsilon \textnormal{ and } \tr{\xi}\leq 1\}$.
	 Let $\alpha\in(0,1)\cup (1,\infty)$, and the minimal quantum Renyi divergence between two subnormalised states $\rho$ and $\sigma$ is defined~\cite{Tomamichel2016QuantumInformation} as $\trD{\rho}{\sigma} \triangleq \frac{1}{\alpha-1}\log\frac{\|\sigma^{\frac{1-\alpha}{2 \alpha}}\rho\sigma^{\frac{1-\alpha}{2 \alpha}}\|^\alpha_\alpha}{\tr{\rho}}$ if $(\alpha<1 \wedge \rho \not\perp \sigma) \vee \rho \ll \sigma$.
	 We denote $\lim_{\alpha \rightarrow 1}\trD{\rho}{\sigma}$ by $ \D{\rho}{\sigma}$ and $\frac{d}{d\alpha}\trD{\rho}{\sigma}|_{\alpha=1}$ by $ \frac{1}{2\log e}\V{\rho}{\sigma}$.
	 For $\alpha \geq 0$ and $\rho_{AB} \in \mathcal{D}_\circ(AB)$, we define $\widetilde{\mathbb{H}}_\alpha^{\uparrow}(A|B)_{\rho}=\sup_{\sigma_B\in\mathcal{D}_\circ(B)}-\trD{\rho_{AB}}{I_A\otimes\sigma_B}$.
	 $\widetilde{\mathbb{H}}_\alpha^{\uparrow}(A)_{\rho}=-\trD{\rho_{A}}{I_A}$ and $\vonH{A}_{\rho}=\lim_{\alpha \rightarrow 1}\widetilde{\mathbb{H}}_\alpha^{\uparrow}(A)_{\rho}$
	 The quantum mutual information with a state $\rho_{AB}$ is defined as $\avgI{A;B}_{\rho}\triangleq \vonH{A}_{\rho} + \vonH{B}_{\rho} - \vonH{AB}_{\rho}$.
	 The min-entropy of $A$ conditioned on $B$ of the subnormalized state $\rho_{AB}$ is defined as $\smH[]{A|B}_\rho\triangleq\sup_{\sigma_B\in\mathcal{D}_\bullet(B)}\sup\{\lambda \in \mathbb{R}: \rho_{AB}\leq \exp(-\lambda)I_A\otimes\sigma_B \}$.
	 The smooth min-entropy of $\rho_{AB}$ is defined as $\smH{A|B}_\rho\triangleq\max_{\tilde{\rho}_{AB}\in\ball[\epsilon]{\rho_{AB}}}\smH[]{A|B}_{\tilde{\rho}}$.
	 And $\sMH[]{A|B}_\rho\triangleq\max_{\sigma_B\in\mathcal{D}_\bullet(B)}\log F(\rho_{AB},I_A \otimes\sigma_B )$.
	 The Holevo information of a quantum channel $\mathcal{N}$ from system $A$ to system $B$ is defined as $\chi{(\mathcal{N}_{A \rightarrow B})} \triangleq {\sup}_{\rho_{XA}} \I XB$, where $\rho_{XA}$ is a CQ-state.
	 The Holevo information of an emsemble is defined as $\chi{\{p(x),\sum_x p(x) \rho_x\}} \triangleq \vonH{\sum_x p(x) \rho_x} - \sum_x p(x)\vonH{ \rho_x}$.
	 For $a,b \in \mathbb{N}$ such that $a \leq b$, we define $[a;b] \triangleq \{a,a+1,\dots,b-1,b\}$.

\section{System Model and Protocol}
\label{sec:system-and-protocol}
\subsection{System Model}
\label{sec:system}
We consider a general \ac{QSDC} model, in which the legitimate parties, Alice and Bob, communicate reliably and securely through a roundtrip quantum channel with the assistance of a public and authenticated classical channel. The roundtrip communication consists of three parts: \begin{inparaenum}[1)] \item Bob prepares and sends random qubits to Alice via the forward channel, described as a \cmnt{\ac{CPTP}} map $\calL_{B \to A}$ from Bob's system $B$ to Alice's system $A$; \item Alice and Bob estimate the forward channel. If the channel permits, Alice then securely encodes a message $w \in \calW$ with a set of encoders; \item The sequence of encoded qubits is sent to Bob and tapped by Eve via the backward channel, $\calL_{X \to Y Z}$, mapping from system $X$ to Bob's system $Y$ and Eve's system $Z$. Bob estimates the message with a set of decoders while Eve creates information leakage by \emph{collective} attacks. \end{inparaenum} The metrics for reliability and security of a \ac{QSDC} model are defined as follows.

\begin{definition}
  \label{def:protocol_reliable_secure}
  \sloppy A \ac{QSDC} model consists of a set of encoders, $\{\calE^{(w)}_{A \to X}\}_{(w)}$, a set of decoders, $\{\Pi_Y^{(w)}\}_{(w)}$, a forward channel $\calL_{B \to A}$, a backward channel $\calL_{X \to YZ}$, and a public, authenticated classical channel. Let $\rho_{W Y Z C}$ denote the joint state comprising the classical message $W$, Bob's received system $Y$, Eve's received system $Z$, and public information $C$. The model is said to be $\epsilon_\textnormal{R}$-reliable and $\epsilon_\textnormal{S}$-secure if $\P{W \neq \hat{W}} \leq \epsilon_\textnormal{R}$ and $\frac{1}{2}\norm[1]{ \rho_{WZC} - \rho_W \otimes \rho_{ZC} } \leq \epsilon_\textnormal{S}$, where $W$ and $\hat{W}=\argmax_{w \in \calW} \tr{\rho_Y \Pi_Y^{(w)}}$ denote the transmitted and estimated messages, respectively.
\end{definition}

The focus of the present paper is on designing tools to analyze the extractable key length across blocks. Since the channel varies across blocks, there may be no closed-form expression for the capacity, and characterizing the corresponding capacity is not our interest.

\subsection{Proposed Protocol}
\label{sec:protocol}
\begin{figure}[ht]
    \centering
    \includegraphics[width=1\linewidth]{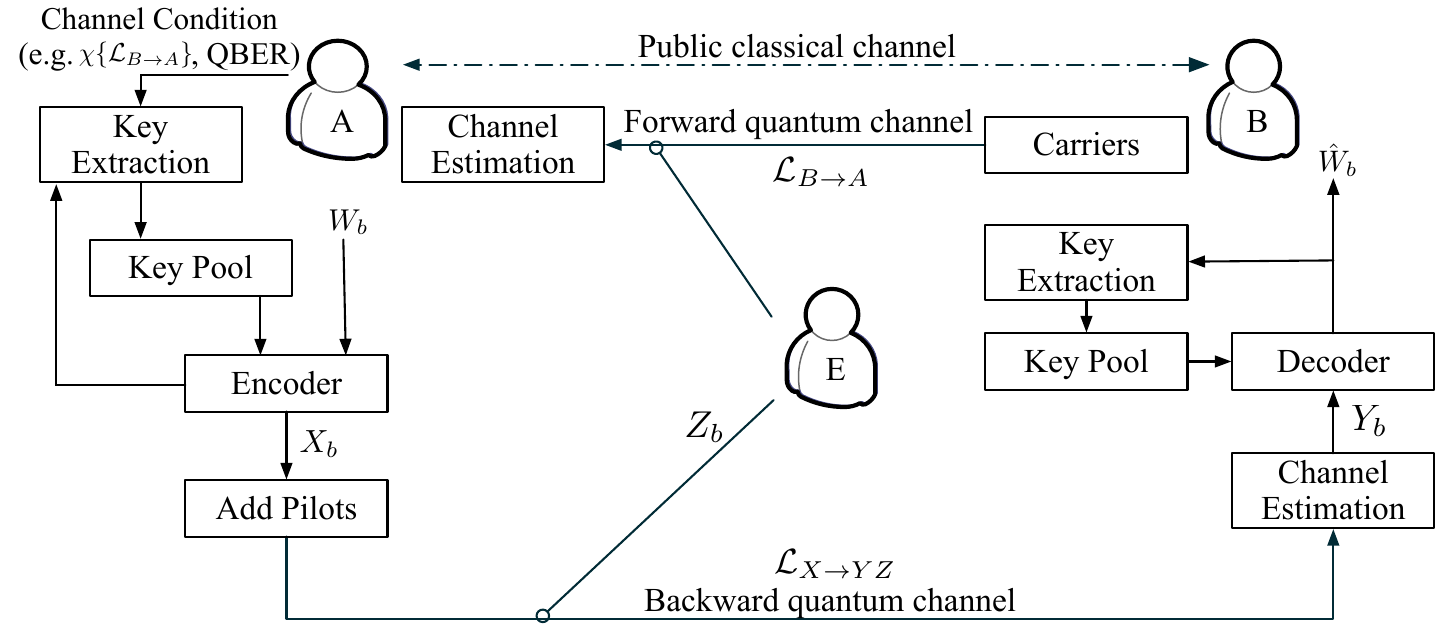}
    \caption{The proposed \ac{QMF}-\ac{QSDC} \cmnt{protocol} in block $b$.}
    \label{fig:block_diag}
\end{figure}

The proposed \ac{QMF}-\ac{QSDC} protocol, adapted from~\cite{Sun2020PracticalQuantum} and illustrated in Fig.~\ref{fig:block_diag}, enables Alice to transmit one message to Bob in a block while generating a shared key and storing it in a key pool for securing messages in subsequent blocks. In each block, Alice applies one-time pads to messages using keys from the key pool, then encodes messages into codewords using a publicly known codebook, and transmits codewords to Bob along with pilot qubits for channel estimation. After Bob estimates the channel and performs error correction on the received qubits, both Alice and Bob apply privacy amplification to extract a key from the coded sequence, where the extractable key length depends on the channel conditions. Unlike two-way \ac{QKD} protocols (e.g.,~\cite{Beaudry2013SecurityTwoway}), our protocol transmits coded sequences rather than random ones. The steps of the proposed protocol are as follows:

\begin{enumerate}
	\item In block $b$, Bob prepares a sequence of $n+p$ qubits, denoted by $B^{n+p}$. Each qubit in the sequence is randomly selected from the set of $\{ \ket{0}, \ket{1}, \ket{+},\ket{-} \}$ with equal probabilities, where $\ket{0}$ and $\ket{1}$ form the computational basis, $\ket{+}= \frac{1}{\sqrt{2}} ( \ket{0} +\ket{1})$, and $\ket{-}= \frac{1}{\sqrt{2}} ( \ket{0} -\ket{1})$. Bob sends the sequence to Alice via the forward channel $\otimes_{i=1}^{n+p} \calL_{B \to A}$.
	\item Eve is allowed to implement a collective attack in each block. She may attach $\tilde{n}$ ancilla systems, denoted by $E^{\tilde{n}}$, to some qubits sent by Bob in Step 1, apply a joint unitary, and resend the qubits $\tr[E]{U \rho_{BE}^{\otimes \tilde{n} } U^\dagger}$ through the channel. Eve attempts to gain information about Alice's transmission by collectively measuring the ancilla qubits afterward.
	\item When Alice receives the sequence of qubits, denoted by $A^{n+p}_b$, she randomly measures $p$ of them in either the $Z$-basis or the $X$-basis to estimate the forward channel, with the aid of the public authenticated classical channel. The remaining $n$ qubits are reserved for message encoding in the next step. Based on the forward channel estimation, Alice and Bob may either abort the protocol or request a retransmission.
	\item Based on the forward and backward channel estimations from previous blocks, Alice and Bob determine the coding rate for the current block $b$, denoted by $R_{\textnormal{code},b}$, and prepare a set of encoders, $\{\calE^{(w, k)}_{A^n \to X^n}\}_{(w,k)}$, with corresponding decoders, $\{\Pi_{Y^n}^{(w,k)}\}_{(w,k)}$. Alice applies a one-time pad to the message $W_b \in \calW$ using a key $K_{b-1}^{\textnormal{mix}}$ drawn from a key pool, then encodes the encrypted message, $L_b=W_b\oplus K_{b-1}^{\textnormal{mix}}$, as a classical codeword on the sequence of qubits received in Step 3. Applying the unitary operator, $U_Y \triangleq \ket{0}\bra{1} - \ket{1}\bra{0}$, represents a classical bit ``1'' while applying the identity operator to leave the qubit unchanged represents a classical bit ``0.'' The encoded sequence of $n$ qubits is then sent back to Bob along with $p$ interlaced pilots through the backward channel. The entire sequence is denoted by ${X^{}}^{n+p}_b$.
	\item Eve may obtain information about the message $W_b$ by tapping the backward channel, acquiring $Z_b$, and performing a joint measurement with the ancilla qubits she previously attached in Step~2.
	\item After Bob receives the sequence of qubits ${Y^{}}^{n+p}_b$, Alice reveals the pilots and their positions. Alice and Bob then jointly perform backward channel estimation over the authenticated public channel. If the channel conditions are poor, they request a retransmission and skip the remaining steps of the current block.
	\item After removing the pilots and obtaining $Y^n_b$, Bob decodes the codeword using the set of decoders, $\{\Pi_{Y^n}^{(w, k)}\}_{(w,k)}$. In this process, Bob performs measurements in the same basis as prepared in Step 1, aided by the key.
	\item Alice and Bob extract a fresh key $K_b$ by applying privacy amplification to the codeword. \cmnt{The extractable key length is determined by the code and channel conditions, for example, the Holevo information of the channel.} The key $K_b$ is stored in a key pool for subsequent use.
\end{enumerate}

Through retransmission and coordination over the authenticated channel, the protocol offers flexibility in replenishing the key pool required in Step 4. Unlike the two-way DL04 protocol~\cite{Deng2004SecureDirect}, the \ac{QMF}-\ac{QSDC} protocol does not require quantum memories but operates over a sequence of dependent blocks. Compared to the existing \ac{QMF}-\ac{QSDC} protocol~\cite{Sun2020PracticalQuantum}, the unique attributes of our proposed protocol are:
\begin{inparaenum}[a)]
\item the analysis accounts for the non-\ac{i.i.d.} sequence of qubits resulting from codeword transmission;
\item no wiretap coding is used and single privacy-amplification-based key extraction mechanism is employed.
\end{inparaenum}

\section{Protocol Analysis:\\ Privacy Amplification on Coded Sequences}
\label{sec:key-rate-analysis}

Since the channel is only estimated \emph{after} each block transmission, coding must be universal and compatible with arbitrary channels. Universal reliability is achieved through appropriate retransmissions upon detecting errors. Secrecy must also be universal and our protocol relies on privacy amplification rather than wiretap codes.

\sloppy Our analysis focuses on characterizing a lower bound for the extractable key length when performing privacy amplification. This section focuses on a single block $b$ hence we drop the index $b$ from the notation. According to the quantum leftover hashing lemma~\cite{Tomamichel2011LeftoverHashing}, the extractable key length in a block of the proposed protocol is determined by the smooth min-entropy of the transmitted codeword, conditioned on the information available to the eavesdropper, denoted by $\smH{X^n|Z^n}$. \cmnt{We assume that Bob can reliably recover $X^n$ from $Y^n$ using the code.} The public information $C^n$ depends on pilots but is independent of the messages. The codeword structure makes the sequence $X^n$ non-\ac{i.i.d.}, requiring a universal bound on extractable key length that depends of code-specific characteristics. To address this, we leverage bounds from entropy accumulation and the quantum \ac{AEP}~\cite{Dupuis2020EntropyAccumulation, Marwah2024SmoothMinentropy, Tomamichel2009FullyQuantum}.

\subsection{Extractable Key Length In The Asymptotic Regime}
We begin by establishing a lower bound on the extractable key length over $n$ channel uses in Theorem~\ref{thm:asymptotic_key_length}, considering the asymptotic regime where $n$ is sufficiently large. In this setting, each block consists of $n$ channel uses, divided into $B_\textnormal{sub}$ messages encoded over $m$ channel uses.

\begin{theorem}
\label{thm:asymptotic_key_length}
Let $B_{\textnormal{sub}}, m \in \mathbb{N}^*$ and  $n = B_{\textnormal{sub}} \cdot m$. Let $R_{\textnormal{code}}$ denote the coding rate. Consider a quantum channel $\mathcal{N}_{X \rightarrow Z}$ from system $X$ to system $Z$, and a state $\rho_{XZ}^n \in \mathcal{D}_\circ(\mathcal{H}_{XZ}^{\otimes n})$ representing the joint quantum state of systems $X$ and $Z$ over $n$ channel uses. For a parameter $\epsilon > 0$ and $n \geq \frac{8}{5}\log \frac{2}{\epsilon^2}$, we have
\begin{multline}
		\smH{X^{m^{\otimes B_{\textnormal{sub}}}}|Z^{m^{\otimes B_{\textnormal{sub}}}}}_{\rho^{m^{\otimes B_{\textnormal{sub}}}}}
	\geq
		n 	\left(
				R_\textnormal{code} - \chi(\mathcal{N}_{X \rightarrow Z})
			\right)\\
			- O(n^{\frac{3}{4}}). \nonumber
\end{multline}
\end{theorem}
\begin{proof}
	Using~\cite[Definition 4 and Theorem 9]{Tomamichel2009FullyQuantum}, let $X^{m^{\otimes B_{\textnormal{sub}}}}$ and $Z^{m^{\otimes B_{\textnormal{sub}}}}$ represent two i.i.d quantum systems across $B_{\textnormal{sub}}$ sub-blocks and both $X^{m}$ and $Z^{m}$ are non-i.i.d. product states. Then,
\begin{align}
	&	\smH{X^{m^{\otimes B_{\textnormal{sub}}}}|Z^{m^{\otimes B_{\textnormal{sub}}}}}_{\rho^{m^{\otimes B_{\textnormal{sub}}}}}\nonumber\\
	&\geq
		B_{\textnormal{sub}} \cdot \vonH{X^m|Z^m}_{\rho^m}
	\nonumber\\
	& \quad - 4 \log
		\Bigl(
		 	\sqrt{2^{-\smH[]{{X^m}|{Z^m}}_{\rho^m}}}
			\!+\!\sqrt{2^{\sMH[]{{X^m}|{Z^m}}_{\rho^m}}}
		+1
		\Bigr)
	\nonumber\\
	&\qquad
		\times \sqrt{B_{\textnormal{sub}}} \sqrt{\log\frac{2}{\epsilon^2}} \label{eq:where_T9} \\
	&\geq
		B_{\textnormal{sub}} \cdot m \Bigl(
						R_\textnormal{code} -\chi\left(\mathcal{N}_{X \rightarrow Z}\right)
					\Bigr)
		- 4 \sqrt{B_{\textnormal{sub}}}\log \left( 2^{\frac{mR_\textnormal{code}}{2}}+2 \right)
		\nonumber\\
	&\quad
		\times \sqrt{\log\frac{2}{\epsilon^2}} \label{eq:where_single_letter_I}\\
	&\geq
		n 	\Bigl(
				R_\textnormal{code} - \chi \left(\mathcal{N}_{X \rightarrow Z}\right)
			 \Bigr)
		- O \left(n^{\frac{3}{4}}\right) \label{eq:where_Hmin},
\end{align}
\sloppy where~\eqref{eq:where_T9} follows by~\cite[Theorem 9]{Tomamichel2009FullyQuantum},~\eqref{eq:where_single_letter_I} follows by $\vonH{X^m|Z^m}_{\rho^m} =\vonH{X^m}_{\rho^m}-\I{X^m}{Z^m}_{\rho^m} \geq \sum_i^m \vonH{X_i}_{\rho_i} - \I{X_i}{Z_i}_{\rho_i} \geq mR_\textnormal{code}-\sum_i^m\mathbb{I}_\textnormal{acc}({X_i};{Z_i})_{\rho_i}\geq m(R_\textnormal{code}-\chi(\mathcal{N}_{X \rightarrow Z}))$ with the memoryless channel assumption, and $\sMH[]{X^m|Z^m}_{\rho^m} \leq \sMH[]{X^m}_{\rho^m} = mR_\textnormal{code}$.
\end{proof}

Theorem~\ref{thm:asymptotic_key_length} provides insight into how the key extraction performs asymptotically as a function of code rate. However, assuming the number of sub-blocks is large may be problematic in practice. When the message is segmented into multiple parts, the corresponding codeword length decreases, leading to degraded coding performance. Furthermore, the scaling in {\small$O(n^{\frac{3}{4}})$} is undesirable and likely sub-optimal. 

\subsection{Extractable Key Length In The Finite-Length  Regime}
\label{sec:finite_key_length}
 To address the aforementioned issues, we establish a non-asymptotic lower bound on the extractable key length in Theorem~\ref{thm:non_asymptotic_key_length}, analyzing the finite-length coding regime. The bound is derived from a general one-shot smooth min-entropy result, refined using a Taylor expansion of the R\'enyi divergence and expressed in terms of coding and channel parameters.

\begin{theorem}
\label{thm:non_asymptotic_key_length}
	Given a quantum channel $\calN_{X \rightarrow Z}$, let $\rho_{X^nZ^n}=\sum_{i \in M} p(i)\ket{i}\bra{i}_{X^n}\otimes \rho_{Z^n}^{(i)}$.
	Define $\sigma_{Z^n} \triangleq \otimes_{j=1}^n \left( \sum_{i=1}^M \frac{1}{M}\rho_j^{(i)} \right)=\otimes_{j=1}^n \tilde{\rho_j}$.
	For $1 > \epsilon > 0$ and a constant $C\in\mathbb{R}$, we have
\begin{align}
	&	\smH{X^n|Z^n}_{\rho_{X^nZ^n}}
		\geq
		n R_\textnormal{code} - n\chi{(\mathcal{N}_{X \rightarrow Z})} \nonumber\\
	& \qquad \qquad
		-\frac{\V{\rho_{X^nZ^n}}{I_{X^n} \otimes \sigma_{Z^n}}}{2 n \log e} 
		-\frac{C}{n^2}-n g(\epsilon) \nonumber,
\end{align}
where
\begin{align*}
		\V{\rho_{X^nZ^n}}{I_{X^n} \otimes \sigma_{Z^n}}
	&= 
		\sum_{i=1}^M \frac{1}{M} \sum_{j=1}^n \V{\rho_j^{(i)}}{\tilde{\rho}_j}
	\\
	& \quad
		+\sum_{i=1}^M \frac{1}{M} \left[\sum_{j=1}^n \D{\rho_j^{(i)}}{\tilde{\rho}_j}\right]^2
	\\
	& \quad
	-\left[\sum_{i=1}^M \frac{1}{M} \sum_{j=1}^n \D{\rho_j^{(i)}}{\tilde{\rho}_j}\right]^2,
\end{align*}
and $g(\epsilon) \triangleq -\log(1-\sqrt{1-{\epsilon}^2})$.
\end{theorem}

\begin{proof}
Using~\cite[(6.96), (6.99)]{Tomamichel2016QuantumInformation}, we obtain the one-shot bound
\begin{align}
	\smH{X|Z}_{\rho_{XZ}}
	& \geq
			\rH{X|Z}_{\rho_{XZ}}-\frac{g(\epsilon)}{\alpha-1} \\
	& = 	
			\sup_{\sigma_Z} -\trD{\rho_{XZ}}{ I_{X}\otimes \sigma_{Z}}
			-\frac{g(\epsilon)}{\alpha-1} \label{eq:def_H} \\
	& \geq 	
			-\D{\rho_{XZ}}{ I_{X}\otimes \sigma_{Z}} \nonumber\\
	& \quad	
			-\frac{\left( \alpha-1 \right)}{2\log e}\V{\rho_{XZ}}{ I_{X}\otimes \sigma_{Z}}\nonumber\\
  	& \quad	-\left( \alpha-1 \right)^2 C -\frac{g(\epsilon)}{\alpha-1} \label{eq:second_order_D},
\end{align}
where~\eqref{eq:def_H} follows by one definition of $\mathbb{H}_\alpha$, and \eqref{eq:second_order_D} follows by the Taylor series expansion for $\widetilde{\mathbb{D}}_\alpha$. The bound over $n$ channel uses is derived by substituting the $\mathbb{D}$ and $\mathbb{V}$ terms in~\eqref{eq:second_order_D} with Lemma.~\ref{lem:single_letterize_D} and setting $\alpha = 1 + \frac{1}{n}$.
\end{proof}

\begin{lemma}
\label{lem:single_letterize_D}
	Let $\rho_{X^nZ^n}=\sum_{i=1}^M \frac{1}{M} \ket{i}\bra{i}_{X^n} \otimes \left( \otimes_{j=1}^n \rho^{(i)}_j \right)_{Z^n}$ which represents the joint state of the codeword $X^n$ and the eavesdropper’s observation $Z^n$, where $i$ denotes message indices and $j$ denotes the bit indices within codewords. Let $R_\textnormal{code}$ denote the coding rate, $\frac{\log M}{n}$. Define $\sigma_{Z^n} \triangleq \otimes_{j=1}^n \left( \sum_{i=1}^M \frac{1}{M}\rho_j^{(i)} \right)$. Then,
\begin{align}
	&\D{\rho_{X^nZ^n}}{I_{X^n} \otimes \sigma_{Z^n}} 
		\leq
		-n R_\textnormal{code} +n \; \chi\{ \frac{1}{M}, \rho^{(i)}\}, \nonumber\\
	&\V{\rho_{X^nZ^n}}{I_{X^n} \otimes \sigma_{Z^n}} 
		= 	\sum_{i=1}^M 	\frac{1}{M}
						\sum_{j=1}^n\V{\rho_j^{(i)}}{\sum_{\ell=1}^M \frac{1}{M} \rho_j^{(\ell)}} \nonumber\\
	& \qquad \qquad \qquad
			+ \sum_{i=1}^M \frac{1}{M}
						\left[
						\sum_{j=1}^n
						\D{\rho_{j}^{(i)}}
						{\sum_{\ell=1}^M \frac{1}{M} \rho_{j}^{(\ell)}}\right]^2
						\nonumber\\
	& \qquad \qquad \qquad
		 	-\left[
				\sum_{i=1}^M \frac{1}{M} \sum_{j=1}^n
				\D{\rho_{j}^{(i)}}
				  {\sum_{\ell=1}^M \frac{1}{M}\rho_{j}^{(\ell)}}
		\right]^2 \nonumber.	
	\end{align}
\end{lemma}
\begin{proof}

\begin{align}
  	&\D{\rho_{X^nZ^n}}{ I_{X^n} \otimes \sigma_{Z^n}} \nonumber\\
  	&= 	\tr{\rho_{X^nZ^n} \biggl[\log(\rho_{X^nZ^n})-\log\left( I_{A^n}\otimes\sigma_{Z^n} \right)   \biggr]}
  	\nonumber \\
  	&= 	\sum_{i=1}^M
  		\operatorname{tr}
  			\biggl(
  			\left(
  				\frac{1}{M} \otimes_{j=1}^n \rho^{(i)}_j 	\right) \nonumber\\
  	&	\quad
  			\times \left[
  					\log \left(  \frac{1}{M} \otimes_{j=1}^n \rho^{(i)}_j 	\right)
  		 		-
  		 			\log \left(  \otimes_{j=1}^n \left( \sum_{\ell=1}^M \frac{1}{M}\rho_j^{(\ell)} \right) \right)
  		\right]
  		\biggr)
  	\nonumber\\
  	&=	\sum_{i=1}^M \D{\frac{1}{M} \otimes_{j=1}^n \rho^{(i)}_j}{\otimes_{j=1}^n \left( \sum_{\ell=1}^M \frac{1}{M}\rho_j^{(\ell)} \right)}
  	\nonumber\\
  	&= 	\sum_{i=1}^M
  			\biggl( 
  			-\frac{1}{M} \log M
  	\nonumber\\
  	&	\quad \qquad  \qquad
  		+	\frac{1}{M}
  			\D{\otimes_{j=1}^n \rho^{(i)}_j}{\otimes_{j=1}^n \left( \sum_{\ell=1}^M \frac{1}{M}\rho_j^{(\ell)}\right)}
  			\biggr) 
  	\label{eq:apply_lem_factor_out_D}
  	\\
  	&= -\log M + \sum_{i=1}^M \frac{1}{M} \sum_{j=1}^n
  		\operatorname{tr} \biggl( \rho_j^{(i)} \log \rho_j^{(i)}
  	\nonumber\\
  	&	\quad \qquad  \qquad  \qquad  \qquad  \qquad  \qquad
  		-\rho_j^{(i)} \log \left( \sum_{\ell=1}^M \frac{1}{M}\rho_j^{(\ell)} \right) 
  		\biggr)
  	\nonumber\\
  	&= 	-nR_\textnormal{code} +\sum_{j=1}^n \chi \left\{p_i(i)=\frac{1}{M},\rho_j^{(i)} \right\} \nonumber\\
  	&\leq -n R_\textnormal{code} +n \; \chi{(\mathcal{N}_{X \rightarrow Z})} \nonumber,
\end{align}
where \eqref{eq:apply_lem_factor_out_D} follows by Lemma~\ref{lem:factor_out_D} in~\cite{Su2026QuantumMemoryFreeQuantum}. For the detailed derivation of $\V{\rho_{X^nZ^n}}{I_{X^n} \otimes \sigma_{Z^n}}$, refer to~\cite{Su2026QuantumMemoryFreeQuantum}, Appendix~\ref{sec:proof_lem_single_letterize_D}.
\end{proof}

Other code characteristics are required for further bounding $\smH{X^n|Z^n}_{\rho_{X^nZ^n}}$. As an example, we specialize Theorem~\ref{thm:non_asymptotic_key_length} to binary linear codes in Corollary~\ref{cor:special_code_general}, for which codewords yield an average uniform distribution over $\{0,1\}$ in each bit position.

\begin{corollary}
\label{cor:special_code_general}
	Given a linear code $\mathcal{C}$, assume that 
$\mathbb{E}_{\mathcal{C}}[ \rho_j^{(i)} ] \cmnt{=} \tilde{\rho}$
for every bit position $j$. Then, 
	\begin{align}
		&\smH{X^n|Z^n}_{\rho_{X^nZ^n}} 
 			\geq
 				 nR_\textnormal{code}
 			  	-n \chi{(\mathcal{N}_{X \rightarrow Z})}
 			  	-n g(\epsilon)
 		\nonumber \\
 		& \quad
  		      	-\frac{1}{2 n \log e}
  				\Biggl( 
					\V{\rho^{0}}{\tilde{\rho}} n
					\nonumber \\
 		& \qquad \qquad \qquad \quad
					+ \Bigl( \V{\rho^{1}}{\tilde{\rho}}-\V{\rho^{0}}{\tilde{\rho}} \Bigr)
					\mathbb{E}_{\mathcal{C}}\Bigl(\operatorname{wt}(x)\Bigr)
				\Biggr)
		\nonumber \\
 		& \quad
  		      	-\frac{1}{2 n \log e}
  		      	\Bigl(
					\D{\rho^1}{\tilde{\rho}}-\D{\rho^0}{\tilde{\rho}}
				\Bigr)^2
					\operatorname{var}\Bigl(\operatorname{wt}(x)\Bigr)
  		\nonumber \\
 		& \quad
  				-\frac{C}{n^2},
	\end{align} where $g(\epsilon) \triangleq -\log(1-\sqrt{1-{\epsilon}^2})$, $\operatorname{var}$ is the variance function, and $\operatorname{wt}(x)$ denotes the weight of codeword $x \in \mathcal{C}$.
\end{corollary}
\begin{proof}
Specializing the $\mathbb{V}$ term in Theorem~\ref{thm:non_asymptotic_key_length}, we have
\begin{align}
	&\V{\rho_{X^nZ^n}}{I_{X^n} \otimes \sigma_{Z^n}}
	\nonumber\\
	&= 	
		n \mathbb{E}_{\mathcal{C}} \left[  \V{\rho^{(i)}}{\tilde{\rho}} \right]
		+\Bigl(\D{\rho^1}{\tilde{\rho}}-\D{\rho^0}{\tilde{\rho}}\Bigr)^2
			\operatorname{var}\Bigl(\operatorname{wt}(x)\Bigr) \label{eq:wt_general}
	\\
	&= 
			\V{\rho^{0}}{\tilde{\rho}} n
		+\biggl(
			\V{\rho^{1}}{\tilde{\rho}}-\V{\rho^{0}}{\tilde{\rho}}
		 \biggr)
				\mathbb{E}_{\mathcal{C}}\Bigl({\operatorname{wt}(x)\Bigr)}
	\nonumber\\
	& 	\qquad
		+
		\Bigl(
			\D{\rho^1}{\tilde{\rho}}-\D{\rho^0}{\tilde{\rho}}
		\Bigr)^2
			\operatorname{var}\Bigl(\operatorname{wt}(x)\Bigr) \nonumber,
\end{align}
where~\eqref{eq:wt_general} follows by $\sum_{j=1}^n \D{\rho_j^{(i)}}{\tilde{\rho}}=\D{\rho^0}{\tilde{\rho}}n+\Bigl(\D{\rho^1}{\tilde{\rho}} - \D{\rho^0}{\tilde{\rho}}\Bigr) \operatorname{wt}(x_i)$. 
\end{proof}

The lower bound on the extractable key length in Corollary~\ref{cor:special_code_general} can be decomposed into three parts. The first part, comprising the code rate, the Holevo information term, and the $g$ function, grows linearly with $n$. The second part, involving the $\mathbb{V}$ terms, remains constant provided that the expected code weight is linear in $n$ and appropriate code design is employed to control the difference in relative entropy variance. The third part, which depends on the $\mathbb{D}$ terms and the variance of the code weight, is more challenging to characterize. To provide an approximate estimate, consider a $(128,64)$ polar code (equivalently, a Reed–Muller code) analyzed in~\cite{Yao2024DeterministicAlgorithm}. In this case, we have $\frac{\operatorname{var}(\operatorname{wt}(x))}{2 n \log e} \approx 0.067$ as the penalty factor. For large $n$, we expect the first part to dominate, making the extractable key length sufficiently large.

Corollary~\ref{cor:special_code_unitary} provides a further simplification by exploiting the property of the proposed protocol that each bit in the codewords is encoded using unitary operators.

\begin{corollary}
\label{cor:special_code_unitary}
	Given a unitary encoder, and a linear code $\mathcal{C}$, assume that 
$\mathbb{E}_{\mathcal{C}}[ \rho_j^{(i)} ] = \frac{1}{2}\rho^0 + \frac{1}{2}\rho^1 = \tilde{\rho}$
for every bit position $j$. Then, it follows that:
	\begin{align}
	&\smH{X^n|Z^n}_{\rho_{X^nZ^n}} 
		\geq
 			nR_\textnormal{code}
 			-n \chi{(\mathcal{N}_{X \rightarrow Z})}
 			-n g(\epsilon)
  			\nonumber\\
  	& \qquad \qquad \qquad \qquad \qquad
  			-\frac{1}{2 \log e}
  				\V{\rho^{0}}{\tilde{\rho}}
  			-\frac{C}{n^2},
	\end{align} where $g(\epsilon) \triangleq -\log(1-\sqrt{1-{\epsilon}^2})$, $\operatorname{var}$ is the variance function, and $\operatorname{wt}(x)$ denotes the weight of codeword $x \in \mathcal{C}$.
\end{corollary}
\begin{proof}
Given $U_Y \triangleq |0\rangle\langle 1|-|1\rangle\langle 0|$ and $\tilde{\rho} \triangleq \frac{1}{2}\left(\rho^0+\rho^1\right)$, we note that
\begin{align*}
	U_Y \; \widetilde{\rho} \; U_Y^{\dagger}
	=\frac{1}{2}\left(U_Y \rho^0 U_Y^{\dagger}+U_Y \rho^1 U_Y^{\dagger}\right)
	=\frac{1}{2}\left(\rho^1+\rho^0\right)=\widetilde{\rho}.
\end{align*}
We first prove the invariance property of relative entropy under unitary. Let $U$ be a unitary, $\rho , \sigma \in \mathcal{D}_\circ(\mathcal{H})$, $\rho' = U \rho U^\dagger$ and $\sigma' = U \sigma U^\dagger$, we have
$\log\rho'=U(\log\rho)U^\dagger$ and $\log\sigma'=U(\log\sigma)U^\dagger$.

Hence,
\begin{align*}
	\mathbb{D}(\rho'|| \sigma')
	&=\operatorname{Tr}\big[U\rho U^\dagger\big(U\log\rho U^\dagger - U\log\sigma U^\dagger\big)\big] \\
 	&=\operatorname{Tr}\big[\rho(\log\rho-\log\sigma)\big]\\
 	&=\mathbb{D}(\rho|| \sigma).
\end{align*}

We then prove the invariance property of relative entropy variance under unitary.
Let $\Delta_{\rho,\sigma}\triangleq \log\rho-\log\sigma$, we have
$\Delta_{\rho',\sigma'}=U(\log\rho-\log\sigma)U^\dagger=U\Delta_{\rho,\sigma}U^\dagger$, and
\begin{align*}
	\big(\Delta_{\rho',\sigma'}-D(\rho'\|\sigma')\big)^2
	=
	U\big(\Delta_{\rho,\sigma}-D(\rho\|\sigma)\big)^2U^\dagger.
\end{align*}
	
Then,
\begin{align*}
	\mathbb V(\rho'\|\sigma')
	&=\operatorname{Tr}\!\Big[\rho'\big(\Delta_{\rho',\sigma'}-D(\rho'\|\sigma')\big)^2\Big]\\
	&=\operatorname{Tr}\!\Big[U\rho U^\dagger \; U(\Delta_{\rho,\sigma}-D(\rho\|\sigma))^2U^\dagger\Big]\\
	&=\mathbb V(\rho\|\sigma).
\end{align*}

Finally, we can show
\begin{align*}
		\D{\rho^{1}}{\tilde{\rho}}
	&=	\D{U_Y\rho^{0}U_Y^\dagger}{U_Y\tilde{\rho}U_Y^\dagger}=\D{\rho^{0}}{\tilde{\rho}}, \\\V{\rho^{1}}{\tilde{\rho}}
	&=	\V{U_Y\rho^{0}U_Y^\dagger}{U_Y\tilde{\rho}U_Y^\dagger}=\V{\rho^{0}}{\tilde{\rho}}.
\end{align*}

The $\V{\rho_{X^nZ^n}}{I_{X^n} \otimes \sigma_{Z^n}}$ term is then replaced with $\V{\rho^{0}}{\tilde{\rho}} n$.
\end{proof}

\section*{Acknowledgment}{Parts of this document have received assistance from generative AI tools to aid in the composition; the authors have reviewed and edited the content as needed and take full responsibility for it.}

\clearpage
\IEEEtriggeratref{22}
\bibliographystyle{IEEEtran}
\bibliography{references.bib}

\clearpage
\appendices
\onecolumn
\section{Supporting Lemmas}
\begin{lemma}
\label{lem:factor_out_D}
For $\rho_1,\sigma_1 \in \mathcal{D}_\circ(\mathcal{H}_1)$, $\rho_2, \sigma_2 \in \mathcal{D}_\circ(\mathcal{H}_2)$, and $a\in \mathbb{R}_+$, we have $\D{a \rho_1 \otimes \rho_2}{\sigma_1 \otimes \sigma_2} = a \log a + a \D{\rho_1 \otimes \rho_2}{\sigma_1 \otimes \sigma_2}$.
\end{lemma}
\begin{proof}
	\begin{align*}
		& \D{a \rho_1 \otimes \rho_2}{\sigma_1 \otimes \sigma_2} \\
		&= \tr{	\left( a \rho_1 \otimes \rho_2 \right)
			   	\left(
			   		 \log \left( a  \rho_1 	 \otimes \rho_2 	\right)
			   		-\log \left( 	\sigma_1 \otimes \sigma_2 	\right)
			   	\right)} \\
		&= a \; \operatorname{tr}\{	\left( \rho_1 \otimes \rho_2 \right)
			   	(
			   		 \log \left( a I_1 \otimes I_2 \right)
			   		+\log \left( \rho_1 \otimes \rho_2 \right) \nonumber\\
		& \quad	   	-\log \left( \sigma_1 \otimes \sigma_2 \right)
			   	)\} \\
		&= a \; \tr{	\left( \rho_1 \otimes \rho_2 \right)} \log a
			+ a \; \operatorname{tr}\{	\left( \rho_1 \otimes \rho_2 \right)
			   	(
			   		 \log \left( \rho_1 \otimes \rho_2 \right) \nonumber \\
		&\quad	 -\log \left( \sigma_1 \otimes \sigma_2 \right)
			   	)\} \\
		&= a \log a + a \D{\rho_1 \otimes \rho_2}{\sigma_1 \otimes \sigma_2}
	\end{align*}	
\end{proof}

\begin{lemma}
\label{lem:factor_out_V}
For $\rho,\sigma \in \mathcal{D}_\circ(\mathcal{H})$, and $a\in \mathbb{R}_+$, we have $\tr{a \rho \left( \log (a \rho) - \log\sigma\right)^2}=a\;\tr{\rho(\log\rho-\log\sigma)^2}+a \log^2 a +2 \left(a\log a \right) \D{\rho}{\sigma}$.
\end{lemma}
\begin{proof}
\begin{align*}
		& \operatorname{tr} \Bigl( a\rho 	\left(\log (a \rho)-\log \sigma \right)
									\left(\log (a \rho)-\log \sigma \right)
							\Bigr) \\
		& =	a \; \operatorname{tr} \Bigl( \rho (
							\log (	a \rho) \log (a \rho)
									+\log\sigma\log\sigma
									-\log\sigma \log(a \rho) \nonumber\\
		&\quad						-\log(a \rho) \log\sigma
							) \Bigr)\\
		& =	a \; \operatorname{tr} \Bigl( \rho \Bigl(
							(\log a + \log \rho) (\log a + \log \rho)
									+\log\sigma\log\sigma \nonumber\\
		& \qquad \qquad \quad
									-\log\sigma (\log a + \log \rho)
									-(\log a + \log \rho) \log\sigma
							\Bigr) \Bigr)\\
		& =	a \; \operatorname{tr} \Bigl( \rho (
								\log \rho \log \rho + \log \sigma \log \sigma
								-\log\sigma\log\rho -\log\rho\log\sigma \nonumber\\
		& \qquad \qquad \quad
							+\log a \log a + 2\log a \log \rho-2\log a \log \sigma
							) \Bigr)\\
		& =	a \; \tr{\rho(\log \rho-\log \sigma)^2}
				+a \log ^2 a + 2  \left( a \log a \right) \D{\rho}{\sigma}
\end{align*}	
\end{proof}

\section{Proof of Lemma~\ref{lem:single_letterize_D}}
\label{sec:proof_lem_single_letterize_D}

\begin{align}
  	&\V{\rho_{X^nZ^n}}{I_{X^n} \otimes \sigma_{Z^n}} \nonumber\\
  	&= \tr{\rho_{X^nZ^n} \Bigl( \log (\rho_{X^nZ^n}) - \log(I_{X^n} \otimes \sigma_{Z^n}) \Bigr)^2}
  	-\Bigl(( \D{\rho_{X^nZ^n}}{I_{X^n} \otimes \sigma_{Z^n}} \Bigr)^2 \label{eq:V_two_terms},\end{align}
  	and we calculate the first term of~\eqref{eq:V_two_terms} below.
\begin{align}
	&	\tr{\rho_{X^nZ^n} \Bigl( \log (\rho_{X^nZ^n}) - \log(I_{X^n} \otimes \sigma_{Z^n}) \Bigr)^2}
	\\
	&= \operatorname{tr} \Biggl( \sum_{i=1}^M\frac{1}{M}\ket{i}\bra{i}\otimes \left( \otimes_{j=1}^n \rho_j^{(i)} \right) \nonumber\\
	& \qquad \qquad
		\Biggl(
			\log
				\biggl(
				\sum_{i=1}^M\frac{1}{M}\ket{i}\bra{i}\otimes \left( \otimes_{j=1}^n \rho_j^{(i)} \right)
				\biggr)
			\biggr) 
		-	\log
			\biggl(
				\otimes_{j=1}^n \biggl( \sum_{\ell=1}^M\frac{1}{M}\rho_j^{(\ell)} \biggr) \biggr)
			\biggr)^2
		\Biggr)
	\\
	&= 	\sum_{i=1}^M \operatorname{tr}
		\Biggl(
		\left( \frac{1}{M} \otimes_{j=1}^n \rho_j^{(i)} \right)
		\Biggl(
			\log
			\biggl(
				\frac{1}{M} \otimes_{j=1}^n \rho_j^{(i)}
			\biggr) 
		-	\log
			\biggl(
				\otimes_{j=1}^n \sum_{\ell=1}^M \frac{1}{M} \rho_j^{(\ell)}
			\biggr)
		\Biggr)^2
		\Biggr) 
	\\
	&= \sum_{i=1}^M
		\Biggl[ 
			\frac{1}{M}\operatorname{tr}
			\Biggl(
				\left(  \otimes_{j=1}^n \rho_j^{(i)} \right)
				\biggl(
					\log
						\Bigl( \otimes_{j=1}^n \rho_j^{(i)} \Bigr) 
					-\log
						\Bigl(
							\otimes_{j=1}^n \sum_{\ell=1}^M \frac{1}{M} \rho_j^{(\ell)}
						\Bigr)
				\biggr)^2 
			\Biggr)	\nonumber\\
	&\qquad \qquad
			+\frac{1}{M}\log^2M
			-\frac{2}{M}\log M
			\D{\otimes_{j=1}^n \rho_j^{(i)}}
			{\otimes_{j=1}^n \sum_{\ell=1}^M \frac{1}{M} \rho_j^{(\ell)}}
		\Biggr] \label{eq:apply_lem_factor_out_V}
	\\
	&= 					\sum_{i=1}^M 	
						\frac{1}{M}
						\Bigg[ \left(
						\sum_{j=1}^n\V{\rho_j^{(i)}}{\sum_{\ell=1}^M \frac{1}{M} \rho_j^{(\ell)}} \right) 
		+\D{\otimes_{j=1}^n \rho_j^{(i)}}
		   {\otimes_{j=1}^n \sum_{\ell=1}^M \frac{1}{M} \rho_j^{(\ell)}}^2 
	\nonumber\\
	& \qquad \qquad \quad
		+\log^2M - 2 \log M \D{\otimes_{j=1}^n \rho_j^{(i)}}
					 {\otimes_{j=1}^n \sum_{\ell=1}^M \frac{1}{M} \rho_j^{(\ell)}}\Bigg],
\end{align}
where \eqref{eq:apply_lem_factor_out_V} follows by Lemma~\ref{lem:factor_out_V}. And we conclude the calculation of $\mathbb{V}$ by summing the second term of~\eqref{eq:V_two_terms}.
\begin{align}
  	&	\V{\rho_{X^nZ^n}}{I_{X^n} \otimes \sigma_{Z^n}} \\
  	&=
  		\tr{\rho_{X^nZ^n} \left( \log\rho_{X^nZ^n} - \log(I_{X^n} \otimes \sigma_{Z^n}) \right)^2} 
  		-\left( \D{\rho_{X^nZ^n}}{I_{X^n} \otimes \sigma_{Z^n}} \right)^2 \\
  	&=
  		\sum_{i=1}^M
  		\frac{1}{M}
		\Bigg[
			\left(
				\sum_{j=1}^n\V{\rho_j^{(i)}}{\sum_{\ell=1}^M \frac{1}{M} \rho_j^{(\ell)}}
			\right)
		+\D{\otimes_{j=1}^n \rho_j^{(i)}}
					   {\otimes_{j=1}^n \sum_{\ell=1}^M \frac{1}{M} \rho_j^{(\ell)}}^2  
	\nonumber\\
	& \qquad \qquad \quad
					+\log^2M - 2 \log M \D{\otimes_{j=1}^n \rho_j^{(i)}}
						{\otimes_{j=1}^n \sum_{\ell=1}^M \frac{1}{M} \rho_j^{(\ell)}}
		\Bigg]
	\nonumber\\
	& \qquad \;
		-
		\left(
			-\log M + \sum_{i=1}^M \frac{1}{M} \sum_{j=1}^n \D{\rho_{j}^{(i)}}{\left( \sum_{\ell=1}^M \frac{1}{M}\rho_{j}^{(\ell)} \right)}
		\right)^2 \\
	&=
		\sum_{i=1}^M
		\frac{1}{M}
		\Bigg[
			\left(
				\sum_{j=1}^n\V{\rho_j^{(i)}}{\sum_{\ell=1}^M \frac{1}{M} \rho_j^{(\ell)}}
			\right) 
			+\D{\otimes_{j=1}^n \rho_j^{(i)}}
						{\otimes_{j=1}^n \sum_{\ell=1}^M \frac{1}{M} \rho_j^{(\ell)}}^2
		\Bigg]
	\nonumber\\
	& \qquad \;
		-
		\left[
				\sum_{i=}^M \frac{1}{M} \sum_{j=1}^n
				\D{\rho_{j}^{(i)}}
				{\left( \sum_{\ell=1}^M \frac{1}{M}\rho_{j}^{(\ell)} \right)}
		\right]^2 \\
	&=					
		\sum_{i=1}^M
		\frac{1}{M}
		\sum_{j=1}^n\V{\rho_j^{(i)}}{\sum_{\ell=1}^M \frac{1}{M} \rho_j^{(\ell)}}
	   +\sum_{i=1}^M \frac{1}{M}
		\left[
			\sum_{j=1}^n
			\D{\rho_{j}^{(i)}}
			  {\sum_{\ell=1}^M \frac{1}{M} \rho_{j}^{(\ell)}}
		\right]^2
	\nonumber\\
	&\quad \quad -\left[
				\sum_{i=1}^M \frac{1}{M} \sum_{j=1}^n
				\D{\rho_{j}^{(i)}}
				{\left( \sum_{\ell=1}^M \frac{1}{M}\rho_{j}^{(\ell)} \right)}
		\right]^2
\end{align}

\end{document}